\documentclass[letterpaper, 10 pt, conference]{ieeeconf}  % Comment this line out if you need a4paper

\IEEEoverridecommandlockouts                              % This command is only needed if 
\usepackage{amsmath} 
\usepackage{amssymb}  
\usepackage{mathbbol}
\usepackage{cite}
\usepackage{wrapfig}
\usepackage{mathrsfs}
\usepackage{url}
\usepackage{array}
\usepackage{amsmath,amsfonts, amssymb,color}
\usepackage{graphicx}
\usepackage{tabularx}

\newtheorem{theorem}{Theorem}

\newtheorem{lemma}{Lemma}

\newtheorem{assumption}{Assumption}

\newtheorem{remark}{Remark}

\usepackage{graphics}
\usepackage{graphicx}
\usepackage{tabularx}
\usepackage{array}
\usepackage{epstopdf}
\usepackage{lipsum}
\usepackage{tablefootnote}
\usepackage{subfigure}

\newcommand\rev[1]{#1}

\newcommand{\mus}{\mu_{\texttt{{S}}}}
\newcommand{\betap}{\beta_{\texttt{P}}}
\newcommand{\betau}{\beta_{\texttt{U}}}
\newcommand{\yint}{y^*_{\mathtt{INT}}}

\newcommand{\cp}{C_{\texttt{P}}}
\newcommand{\cu}{C_{\texttt{U}}}

\newcommand{\zsbar}{z_{\bar{\texttt{S}}}}
\newcommand{\zibar}{z_{\bar{\texttt{I}}}}

\newcommand{\zidag}{z^{\dagger}_{\bar{\texttt{I}}}}
\newcommand{\zsdag}{z^\dagger_{\bar{\mathtt{S}}}}
\newcommand{\yee}{y_{\texttt{EE}}}
\DeclareMathOperator*{\argmin}{arg\,min}

\newcommand{\Stb}{\texttt{S}}
\newcommand{\Itb}{\texttt{I}}
\newcommand{\Ut}{\texttt{U}}
\newcommand{\Pt}{\texttt{P}}

\newcommand{\ignore}[1]{}

\title{\LARGE \bf
Optimal Bayesian Persuasion for Containing SIS Epidemics
\author{Urmee Maitra$^{1}$, Ashish R. Hota$^{1}$, and Philip E. Par\'e$^{2}$% <-this % stops a space
}
 \thanks{$^{1}$U. Maitra and A. R. Hota are with the Department of Electrical Engineering, IIT Kharagpur, Kharagpur, West Bengal, India, 721302. Email: {\tt\small urmeemaitra93@kgpian.iitkgp.ac.in, ahota@ee.iitkgp.ac.in.}}%
\thanks{$^{2}$Philip E. Par\'e is with the Elmore Family School of Electrical and Computer Engineering, Purdue University, USA. Email: {\tt\small philpare@purdue.edu.}}%
\thanks{Research of the authors was supported in part by the National Science Foundation, grant NSF-ECCS 2032258.}
 }

\allowdisplaybreaks
\begin{document}

\maketitle
\thispagestyle{empty}
\pagestyle{empty}

%%%%%%%%%%%%%%%%%%%%%%%%%%%%%%%%%%%%%%%%%%%%%%%%%%%%%%%%%%%%%%%%%%%%%%%%%%%%%%%%
\begin{abstract}
We consider a susceptible-infected-susceptible (SIS) epidemic model in which a large group of individuals decide whether to adopt partially effective protection without being aware of their individual infection status. Each individual receives a signal which conveys noisy information about its infection state, and then decides its action to maximize its expected utility computed using its posterior probability of being infected conditioned on the received signal. We first derive the static signal which minimizes the infection level at the stationary Nash equilibrium under suitable assumptions. We then formulate an optimal control problem to determine the optimal dynamic signal that minimizes the aggregate infection level along the solution trajectory. We compare the performance of the dynamic signaling scheme with the optimal static signaling scheme, and illustrate the advantage of the former through numerical simulations. 
\end{abstract}

%%%%%%%%%%%%%%%%%%%%%%%%%%%%%%%%%%%%%%%%%%%%%%%%%%%%%%%%%%%%%%%%%%%%%%%%%%%%%%%%
\section{Introduction}\label{section:introduction}
Effective containment of spreading processes or epidemics is a challenging problem. Prior works have developed optimal control techniques for containing epidemics via non-pharmaceutical interventions (such as wearing masks or social distancing) \cite{nowzari2016analysis,bliman2021optimal,kohler2021robust}, allocation of limited medical resources including vaccines and testing kits \cite{parino2023model,thul2023stochastic}, and limiting mobility among individuals \cite{niazi2021optimal}. However, implementing the optimal control actions in the above settings is not straightforward since the suggested actions often require a large number of individuals to adhere to them, many of whom may not even be aware of their true infection status. 

Past works have investigated the impacts of decentralized decision-making by a large group of selfish and strategic individuals on epidemic containment in the framework of game theory \cite{huang2022game,maitra2023sis,satapathi2022epidemic,ye2021game,9683594}. In most of these prior works, the agents were assumed to be aware of their true infection status. However, there are epidemics with asymptomatic infections and in some cases, similar symptoms are observed for multiple diseases; both the above characteristics were true for the recent COVID-19 pandemic. In addition, limited availability of testing kits made it difficult for individuals to learn about their true infection status in frequent intervals.

Consequently, individuals often relied on certain smart-phone applications \cite{lewis2021contact} to learn about their infection status. These applications collected information about several health indicators (such as pulse rate, body temperature, and sleep patterns) from the user, and whether the user was in proximity to other infected users to determine its risk of infection. However, the information conveyed by such applications is not free from error, and is potentially biased. Nevertheless, these applications provide a scalable manner in which individuals can be influenced to adopt protection measures and limit the spread of the epidemic. In this paper, we investigate how to design suitable signals to influence the behavior of strategic individuals towards the optimal containment of SIS epidemics in the framework of Bayesian persuasion \cite{kamenica2011bayesian,kamenica2019bayesian}.  

Although the Bayesian persuasion and information design frameworks have been applied in diverse settings, such as routing games \cite{ferguson2024information,zhu2022information}, resource allocation games \cite{paarporn2023strategically}, and modeling deception and privacy \cite{sayin2021bayesian}, there are only a few works that explore their application in the context of epidemics. Early attempts in this direction include \cite{pathak2022scalable} and \cite{chang2023controlling}, neither of which consider any specific compartmental epidemic model. In a recent work \cite{hota2023bayesian} by a subset of authors of this paper, the authors consider a game-theoretic setting where a large population of individuals adopt protection against the SIS epidemic with asymptomatic infections by relying on a noisy signal they receive from the sender. The authors characterize the stationary Nash equilibrium (SNE) of the game when the signaling scheme remains unchanged throughout the duration of the epidemic. 

In this work, we build upon \cite{hota2023bayesian} and investigate the {\it optimal} signaling scheme to minimize the proportion of infected individuals. First we derive the optimal static signal that minimizes the proportion of infected individuals at the SNE under suitable assumptions. We then formulate an optimal control problem which aims to minimize the integral of the time-varying infected proportion subject to constraints that include the SIS epidemic dynamics as well as an evolutionary learning dynamics that capture the evolution of the actions by the individuals. We numerically illustrate the trajectories of infected proportion obtained under the optimal static and dynamic signaling schemes. The results show that the trajectory of infected proportion is smaller under the dynamic (time-varying) signaling scheme compared to the optimal static signaling scheme, though the difference reduces with time as the trajectories converge to the corresponding SNE.

%%%%%%%%%%%%%%%%%%%%%%%%%%%%%%%%%%%%%%%%%%%%%%%%%%%%%%%%%%%%%%%%%%%%%%%%%%%%%%%%
%%%%%%%%%%%%%%%%%%%%%%%%%%%%%%%%%%%%%%%%%%%%%%%%%%%%%%%%%%%%%%%%%%%%%%%%%%%%%%%%
\section{SIS Epidemic under Bayesian Persuasion}
\label{sec:epi_model}
In this section, we briefly review the setting and main result from \cite{hota2023bayesian}. We consider a large group of agents, each of whom is either \emph{susceptible} $(\Stb)$ or \emph{infected} $(\Itb)$ at a given point of time, but is unaware of its exact infection state. \rev{The proportion of infected individuals at time $t$ is denoted by $y(t)$.} Each agent receives a noisy signal $x$ from the sender from the set $\{ \bar{\Stb}, \bar{\Itb}\}$. The signaling scheme is governed by two parameters $\mus, \mu_{\texttt{I}} \in [0, 1]$ with 
$$\mus := \mathbb{P}[x = \bar{\Stb} \mid \Stb], \quad \mu_{\Itb} := \mathbb{P}[x = \bar{\Itb} \mid \Itb].$$
The agents are aware of the signaling scheme, i.e., the values of $\mus$ and $\mu_{\Itb}$. After receiving the signal, each agent computes its posterior probability of being infected as
\begin{align*}
    \pi^+[\Itb | x] &= \frac{\mathbb{P}[x | \Itb] y(t)}{\mathbb{P}[x | \Itb] y(t) + \mathbb{P}[x | \Stb] (1-y(t))},
\end{align*}
with $\pi^+[\Stb | x] = 1-\pi^+[\Itb | x]$ and \rev{$y(t)$ acts as the prior.}\footnote{Though each individual agent is not aware of its personal infection status, it is reasonable to assume that the overall fraction of infected agents $y(t)$ is known. Information regarding $y(t)$ is often publicly available, e.g., in terms of test positivity rate among the tested population, which though negligible compared to the overall size of the population, is large in absolute numbers.}
 
Each individual chooses among two available actions: adopting protection and remaining unprotected, denoted by $a \in \{\Pt, \Ut\}$. Adopting protection comes with a cost $\cp > 0$ while an infected individual that remains unprotected incurs cost $\cu > 0$. An infected agent that adopts protection (or remains unprotected) spreads infection with probability $\beta_{\Pt}$ (respectively, $\beta_{\Ut}$), with $0 < \beta_{\Pt} < \beta_{\Ut} < 1$. A susceptible agent, upon adopting protection, reduces its likelihood of becoming infected by a factor $\alpha \in (0, 1)$. Further, $\gamma \in (0, 1)$ represents the recovery rate from the disease. Let $\zsbar(t)$ (respectively, $\zibar(t)$) denote the proportion of agents who remain unprotected among the agents that receive signal $\bar{\Stb}$ (respectively, $\bar{\Itb}$) at time $t$. Consequently, the proportion of infected agents evolves as
\begin{align}
    \dot{y} & = ((1 - y) [\betap + (\betau - \betap) (\zibar \mu_{\Itb} + \zsbar (1-\mu_{\Itb}))] \nonumber
    \\ & \qquad \times[\alpha + (1 - \alpha) (\zsbar \mus + \zibar (1 - \mus))] - \gamma) y \label{eq:sis_protection}
    \\ & =: ((1 - y) \beta_{\mathtt{eff}}(z_{\bar{\Stb}},z_{\bar{\Itb}};\mus) - \gamma) y, \nonumber
\end{align}
where the dependence on $t$ is omitted for better readability. 

\rev{Analogous to the classical SIS epidemic model \cite{mei2017dynamics}, the above dynamics has two equilibrium points: a \emph{disease-free equilibrium}, $y_{\mathtt{DFE}} = 0$ which always exist; and an \emph{endemic equilibrium}, $y_{\mathtt{EE}}(z_{\bar{\Stb}},z_{\bar{\Itb}};\mus) := 1 - \frac{\gamma}{\beta_{\mathtt{eff}}(z_{\bar{\Stb}},z_{\bar{\Itb}};\mus)}$ which exists only when $\beta_{\mathtt{eff}}(z_{\bar{\Stb}},z_{\bar{\Itb}};\mus) > \gamma$, and is strictly positive.} We now impose the following assumptions. 

\begin{assumption}\label{assumption:main}
The signal revealed to infected agents is truthful, i.e., $\mu_{\Itb} = 1$; $C_{\Pt} < C_{\Ut}$, which incentivizes infected agents to adopt protection; and \rev{the recovery rate satisfies $\gamma < \alpha \beta_{\Pt}$.} 
\end{assumption}

The assumption $\mu_{\Itb} = 1$ is motivated by the fact that during an actual pandemic, the signaling scheme will err on the side of being safe at the expense of a larger false alarm rate. \rev{Since $0 < \betap < \betau < 1, \alpha \in (0, 1)$, and $(\mus, \mu_{\mathtt{I}}, \zsbar, \zibar) \in [0, 1]^4$, it follows from \eqref{eq:sis_protection} that $\alpha \betap \leq \beta_{\mathtt{eff}}(z_{\bar{\Stb}},z_{\bar{\Itb}};\mus)$. Thus, the assumption $\gamma < \alpha \betap$ guarantees the existence of the endemic equilibrium $y_{\mathtt{EE}}$ with a nonzero infected proportion.}\footnote{\rev{Since the infected proportion is zero at the disease-free equilibrium, our main objective is to minimize the nonzero infected proportion at the endemic equilibrium.}}  

The expected utility (with respect to posterior $\pi^{+}$) received by an agent upon receiving signal $x \in \{\bar{\Stb}, \bar{\Itb}\}$ and choosing action $a \in \{\Pt, \Ut\}$ is denoted by $U[x, a]$, \rev{and the dependence of $U[x, a]$ on the tuple $(y,z_{\bar{\Stb}},z_{\bar{\Itb}})$ is suppressed for better readability.} Following \cite[Section III]{hota2023bayesian}, the difference in the expected utilities for each $x$ is given by
\begin{align}
    \Delta U[\bar{\Stb}] &= (1 - \alpha) L [\betap + (\betau - \betap) \zibar] y - \cp, \nonumber
    \\ \Delta U[\bar{\Itb}] &= \pi^+ [\Stb | \bar{\Itb}] ((1 - \alpha) L [\betap + (\betau - \betap) \zibar] y - \cu) \nonumber
    \\ &  \qquad + \cu - \cp,
    \label{eq:utility}
\end{align}
where $L > 0$ is the loss incurred by a susceptible agent upon being infected. \rev{A summary of important notations introduced above is provided in Table \ref{tab:notation}.}

\begin{table}[t]
\renewcommand{\arraystretch}{1.2}
    \centering
\caption{\rev{Summary of important notations.}}
\rev{\begin{tabularx}{0.48\textwidth} {|p{2cm}|p{5.7cm}|}
 \hline
 $\mathtt{S}, \mathtt{I}$ & Infection state: Susceptible ($\mathtt{S}$), Infected ($\mathtt{I}$) \\
 \hline
 $x \in \{\bar{\Stb}, \bar{\Itb}\}$ & Signal: Susceptible ($\bar{\Stb}$), Infected ($\bar{\Itb}$) \\
\hline
 $a \in \{\Pt, \Ut\}$ & Available action: adopt protection ($\Pt$) or remain unprotected ($\Ut$) \\
\hline
$\mus \in [0, 1]$ & Probability of an agent receiving signal $\bar{\Stb}$ when it is in state $\Stb$ \\
\hline
$\mu_{\mathtt{I}} \in [0, 1]$ & Probability of an agent receiving signal $\bar{\Itb}$ when it is in state $\Itb$ \\
\hline
$\pi^{+}[\mathtt{I} \mid x]$ & Posterior probability of being infected if the received signal is $x$ \\
\hline
$y(t) \in [0, 1]$ & Proportion of infected individuals at time $t$ \\
\hline
$\zsbar(t)$ ($\zibar(t)$) $\in [0, 1]$ & Proportion of agents that do not adopt protection among the agents that receive signal $\bar{\Stb}$ ($\bar{\Itb}$) \\
\hline
$\betap$ ($\betau$) $\in (0, 1)$ & Infection rate of an infected agent that adopts protection (remains unprotected) \\
\hline
$\alpha \in (0, 1)$ & Reduction of infection risk for a susceptible agent upon adopting protection  \\
\hline
$\gamma \in (0, 1)$ & Recovery rate \\
\hline
$\cp > 0$ & Cost of adopting protection \\
\hline
$\cu >0$ & Cost of remaining unprotected \\
\hline
$L > 0$ & Loss of a susceptible agent when infected\\
\hline
$U[x, a]$ & Expected utility of an agent receiving signal $x$ and choosing action $a$ \\
\hline
$\Delta U[x]$ & Difference in expected utility $U[x, \Pt] - U[x, \Ut]$ \\
\hline
$\beta_{\mathtt{eff}}(z_{\bar{\Stb}},z_{\bar{\Itb}};\mus)$ & Effective infection rate of the SIS dynamics \eqref{eq:sis_protection} \\
\hline
$y_{\mathtt{EE}}(z_{\bar{\Stb}},z_{\bar{\Itb}};\mus)$ & Infected proportion at the endemic equilibrium of \eqref{eq:sis_protection}  \\
\hline
\end{tabularx}}
\label{tab:notation}
\end{table}

\rev{Since the utility of an agent depends on the strategies chosen by all other agents $(z_{\bar{\Stb}},z_{\bar{\Itb}})$, we study their interaction within a game-theoretic framework.} A tuple $(y^\star,z^\star_{\bar{\Stb}},z^\star_{\bar{\Itb}})$ is said to be a stationary Nash equilibrium (SNE) when (i) $y^\star$ is the unique nonzero endemic equilibrium of \eqref{eq:sis_protection} with $\zibar = z^\star_{\bar{\Itb}}$ and $\zsbar = z^\star_{\bar{\Stb}}$, and (ii) the following mixed complementarity conditions are satisfied:
$$ \Delta U[\bar{\Stb}] \perp (0 \leq z^\star_{\bar{\Stb}} \leq 1), \quad \Delta U[\bar{\Itb}] \perp (0 \leq z^\star_{\bar{\Itb}} \leq 1).$$
A more elaborate discussion on the SNE is provided in \cite[Definition 1]{hota2023bayesian}. Before stating the main result on the characterization of SNE, we reproduce the following notation from \cite{hota2023bayesian}. We define
    \begin{align}
    & \yee(z_{\bar{\Stb}},z_{\bar{\Itb}};\mus) := 1 - \frac{\gamma}{\beta_{\mathtt{eff}}(z_{\bar{\Stb}},z_{\bar{\Itb}};\mus)},
    \label{eq:yee}
       \\ & h(\zibar, \mus) := (1 - \alpha) L (\betap + (\betau - \betap) \zibar) y_{\texttt{EE}}(1, \zibar; \mus) \nonumber 
        \\ & \qquad \qquad - \cp, \nonumber
        \\ & g(\zibar, \mus) := y_{\texttt{EE}}(1, \zibar; \mus) (\cu - \cp) \nonumber \\
        & \qquad -(1 - \mus) (1 - y_{\texttt{EE}}(1, \zibar; \mus)) (-h(\zibar, \mus)).
        \label{eq:h,g}
    \end{align}
    The quantity $\mus^{\max}$ is defined in \cite[Equation $(11)$]{hota2023bayesian} as:
\begin{align}\label{eq:mus_max}
   \mus^{\max} := \begin{cases}
        & 0, \quad \text{if} \quad g(0;0) \geq 0,
        \\ & \mus^{\star}, \quad \text{where } g(0,\mus^{\star}) = 0, \mus^{\star} \in (0,1). 
    \end{cases}
\end{align}
Finally, we define
\begin{align*}
    y^*_{\texttt{P}} & := 1 - \frac{\gamma}{\alpha \betap}, \quad y^*_{\texttt{INT}} :=  \frac{\cp}{L(1 - \alpha) \betap},
    \\ \zsbar^{\dagger} & := \frac{\gamma - \alpha \betap (1 - y^*_{\texttt{INT}})}{\betap (1 - \alpha) (1 - y^*_{\texttt{INT}}) \mus},
     \\ \mus^{\min} &:= 1 - \frac{(\betau - \gamma) (\cu - \cp)}{\gamma (\cp - (1 - \alpha) L (\betau - \gamma))}. 
\end{align*}
With the above notations in hand, the characterization of the SNE $(y^\star, \zsbar^\star, \zibar^\star)$ as established in \cite{hota2023bayesian} is stated below.

\begin{theorem}(\!\!\!\cite[Theorem $1$]{hota2023bayesian})
Under Assumption \ref{assumption:main}, we have the following characterization of the SNE: 
\begin{enumerate}
\item $(y^*_{\Pt},0,0)$ is the SNE if and only if $y^*_{\Pt} > \yint$.
\item $(\yint, \zsdag, 0)$ is the SNE if and only if $\zsdag \in (0,1)$ or equivalently,
    \begin{align*}
        & 1 - \frac{\gamma}{\alpha \betap} < \yint < 1 - \frac{\gamma}{\betap (\alpha + (1 - \alpha) \mu_{\Stb})}.
    \end{align*}
\item $(\yee(1,0;\mu_{\Stb}),1,0)$ is the SNE if and only if 
    \begin{align*}
        &  \mu_{\Stb} \in [\mu^{\max}_{\Stb},1], \quad 1 - \frac{\gamma}{\betap (\alpha + (1 - \alpha) \mu_{\Stb})} \leq \yint. 
    \end{align*} 
\item     $(\yee(1,\zidag;\mu_{\Stb}),1,\zidag)$ with $\zidag \in (0,1)$ being the unique value satisfying $g(\zidag,\mu_{\Stb})=0$ is the SNE if and only if $\mu_{\Stb} < \mu^{\max}_{\Stb}$, and either of the following two conditions are satisfied:
    \begin{align*}
        & L(1-\alpha)(\betau - \gamma) < C_{\Pt} \quad \text{with} \quad \mu_{\Stb} >\mu^{\min}_{\Stb}, \quad \text{or} 
        \\ & L(1-\alpha)(\betau - \gamma) > C_{\Pt}.
    \end{align*} 
\item $(1-\frac{\gamma}{\betau},1,1)$ is the SNE if and only if $L(1-\alpha)(\betau - \gamma) < C_{\Pt}$, and $\mu_{\Stb} < \mu^{\min}_{\Stb}$.  
\end{enumerate}
\label{lemma:sne_char}
\end{theorem}

The above result characterizes the SNE for a given set of game parameters and a given value of $\mu_{\Stb}$. 

%%%%%%%%%%%%%%%%%%%%%%%%%%%%%%%%%%%%%%%%%%%%%%%%%%%%%%%%%%%%%%%%%%%%%%%

\section{Optimal Static Signal}
\label{sec:oss}

In this section, \rev{we consider the case where the principal or the sender chooses the signaling scheme in order to minimize the infected proportion at the SNE.} The optimal static signal $\mus$ that achieves the smallest proportion of infected individuals at the SNE is denoted by $\mu^\star_{\texttt{S(stat)}}$. Our analysis holds under the following assumption. 

\begin{assumption}\label{assump:cp_more}
The protection cost $\cp$ satisfies $\cp > (1 - \alpha) L (\betau - \gamma)$.
\end{assumption}

\begin{remark}
Note that Assumption \ref{assump:cp_more} implies
\begin{align*}
        \cp > & (1 - \alpha) L (\betau - \gamma) \implies \frac{\cp}{(1 - \alpha) L \betau} > 1 - \frac{\gamma}{\betau}
      \\ \implies & \underbrace{\frac{\cp}{(1 - \alpha) L \betap}}_{y^*_{\texttt{INT}}} > \frac{\betau}{\betap} - \frac{\gamma}{\betap} > \underbrace{1 - \frac{\gamma}{\alpha \betap}}_{y^*_{\texttt{P}}},
\end{align*}
where the last inequality follows since $\betau > \betap$ and $\alpha \in (0, 1)$. Consequently, the SNE $(y^*_{\texttt{P}},0,0)$ (Case $1$ of Theorem \ref{lemma:sne_char}) does not exist under Assumption \ref{assump:cp_more}. Now, observe that a necessary condition for existence of the SNE $(y^*_{\texttt{{INT}}}, \zsbar^\dagger, 0)$ is 
\begin{align*}
\frac{\cp}{(1 - \alpha) L \betap} & < 1 - \frac{\gamma}{\betap (\alpha + (1 - \alpha) \mus)}
\\ \implies \frac{\cp}{(1 - \alpha) L \betau} & < 1 - \frac{\gamma}{\betau},
\end{align*}
which is not possible under Assumption \ref{assump:cp_more}. Nevertheless, at these two SNEs, the proportions of infected agents are $y^*_{\Pt}$ and $\yint$, both of which do not depend upon $\mus$, and hence there is no notion of an optimal signal. More importantly, the existence of the remaining three SNEs, whose infected proportion depends on $\mus$ is not affected by Assumption \ref{assump:cp_more}. 
\label{remark:assump_cp_more}
\end{remark}
    
We start our derivation of the optimal static signal with the following lemma. 

\begin{lemma}\label{lemma:y_decreasing}
\rev{Suppose Assumptions \ref{assumption:main} and \ref{assump:cp_more} hold.} Let $\mus^{\max} > 0$ and $\mus^{\min}  < \mus^{\max}$. Then, the endemic infection level $y_{\texttt{EE}}(1, \zibar^\dagger; \mus)$ is strictly decreasing in $\mus$ over the range $\mus \in (\mus^{\min}, \mus^{\max})$.
\end{lemma}
\begin{proof}
It follows from Theorem \ref{lemma:sne_char} that when $\mus^{\max} > 0$ and $\mus \in (\mus^{\min},\mus^{\max})$, the SNE is given by $(y_{\texttt{EE}}(1, \zibar^\dagger; \mus),1,\zibar^\dagger)$. In addition, $\cp > (1 - \alpha) L (\betau - \gamma)$ implies $\mus^{\min} < 1$. Further, it can be shown that  $\mus^{\max} \in (0, 1)$. The infected proportion at the SNE is given by
    \begin{align}
        & y_{\texttt{EE}}(1, \zibar^\dagger; \mus) = 1 \label{eq:y_ee} 
        \\ & - \frac{\gamma}{(\betap + (\betau - \betap) \zidag) (\alpha + (1 - \alpha) (\zidag + (1 - \zidag) \mus))}, \nonumber
    \end{align}
where $\zidag \in (0, 1)$ is the unique value satisfying $g(\zidag; \mus) = 0$. As $\mus$ varies in the range $(\mus^{\min}, \mus^{\max})$, $\zidag \in (0, 1)$ varies while preserving $g(\zidag, \mus) = 0$. In the rest of the proof, we denote $\zidag(\mus)$ to indicate that $\zidag$ is a function of $\mus$.

With a slight abuse of notation, we use $y_{\texttt{EE}}(\mus)$ to denote the endemic infection level $y_{\texttt{EE}}(1, \zibar^\dagger; \mus)$ at the SNE. It should be noted that $y_{\texttt{EE}}(\mus)$ is a function of only $\mus$, since $\zidag(\mus)$ itself is a function of $\mus$. We now introduce the following functions:
    \begin{align*}
        \beta_{\zidag}(\mus) &:= \betap + (\betau - \betap) \zidag(\mus),
        \\ \alpha_{\zidag}(\mus) &:= \alpha + (1 - \alpha) (\zidag(\mus) + (1 - \zidag(\mus)) \mus),
        \\ L_{\texttt{eq}}(\mus) &:= (1 - \alpha) L (\betau - \betap) (1 - \mus) (1 - \yee(\mus)),
        \\ w(\mus) &:= \beta_{\zidag}(\mus) \cdot \alpha_{\zidag}(\mus).
    \end{align*}
Note that when $\alpha \in (0, 1)$, $\mus \in (0, 1)$ and $0 < \betap < \betau < 1$, $w(\mus)$ is a non-zero, continuous and differentiable function of $\mus$. Therefore, $\yee(\mus) = 1 - \frac{\gamma}{w(\mus)}$ is also continuous and differentiable in $\mus$. We now introduce the operator $\nabla_{a}(b)$ to denote the derivative of $b$ w.r.t $a$, i.e., $\nabla_{a}(b) := \frac{\text{d}}{\text{d} a} (b)$. By applying the operator $\nabla_{\mus}(\cdot)$ on both sides of the above identity, we obtain 
$\nabla_{\mus}(\yee) = \frac{\gamma}{(w(\mus))^2} \nabla_{\mus}(w)$. Thus, when $w(\mus)$ is strictly decreasing in $\mus$, $\yee(\mus)$ is also strictly decreasing in $\mus$. 

We now derive expressions of $\zidag(\mus)$ and $\nabla_{\mus}(\zidag)$, respectively. By substituting $h(\zidag(\mus), \mus)$ in $g(\zidag(\mus), \mus)$, and equating $g(\zidag(\mus), \mus) = 0$, we obtain
\begin{align}
& (\cu - \cp) \yee(\mus) = (1 - \mus) (1 - \yee(\mus)) \nonumber
    \\ & \times \big[-\!(1 \!-\! \alpha) L (\betap + (\betau \!-\! \betap) \zidag(\mus)) \yee(\mus) + \cp \big] \nonumber
\\ \Rightarrow ~ & L_{\texttt{eq}}(\mus) \zidag(\mus) \yee(\mus) = (1 - \mus) (1 - \yee(\mus)) \cp \nonumber
\\ & - (1 - \alpha) L (1 - \mus) (1 - \yee(\mus)) \betap \yee(\mus) \nonumber
\\ & - \yee(\mus) (\cu - \cp), \nonumber
\\ \Rightarrow ~ & \zidag(\mus) = \frac{\cp w(\mus)}{L (w(\mus)-\gamma) (1 \!-\! \alpha) (\betau - \betap)} -\frac{\cu \!-\! \cp}{L_{\texttt{eq}}(\mus)} \nonumber
\\ & \qquad - \frac{\betap}{\betau - \betap}, \label{eq:zi_new}
\end{align}
where the last implication is obtained by substituting $\yee(\mus) = 1 - \frac{\gamma}{w(\mus)}$ in the preceding equation. The R.H.S. of the above equation is well-defined, and continuous in $\mus$ over the specified range. We now compute $\nabla_{\mus}(\zidag)$ as
\begin{align}
& \frac{\cp}{L (1 - \alpha) (\betau - \betap)} \nabla_{\mus} \Big(\frac{w(\mus)}{w(\mus) - \gamma}\Big) \nonumber
    \\ & -\frac{\cu - \cp}{(1 - \alpha) L (\betau - \betap)} \nabla_{\mus} \Big((1 - \mus)^{-1} (1 - \yee(\mus))^{-1}\Big) \nonumber
    \\ = & \frac{-\gamma \cp \nabla_{\mus} (w(\mus))}{L (1 - \alpha) (\betau - \betap) (w(\mus) - \gamma)^2} \nonumber \\ &- \frac{\cu - \cp}{(1 - \alpha) L (\betau - \betap)} \Big(\frac{(1 - \mus)^{-1}}{(1 - \yee(\mus))^2} \nabla_{\mus} (\yee(\mus)) \nonumber
    \\ &\qquad \quad +\frac{(1 - \yee(\mus))^{-1}}{(1 - \mus)^2}\Big), \nonumber
    \\  = &  -\gamma \nabla_{\mus}(w(\mus)) \Big(\frac{\cp}{L(w(\mus) - \gamma)^2 (1 - \alpha) (\betau - \betap)} \nonumber
        \\ &+ \frac{\cu - \cp}{ \gamma w(\mus) L_{\texttt{eq}}(\mus)}\Big)  - \frac{\cu - \cp}{(1 - \mus)L_{\texttt{eq}}(\mus)}, \label{eq:del_zi_new}
\end{align}
where, the final equation is obtained by substituting $\yee(\mus) = 1 - \frac{\gamma}{w(\mus)}$, $\nabla_{\mus}(\yee) = \frac{\gamma}{(w(\mus))^2} \nabla_{\mus}(w)$, and $L_{\texttt{eq}}(\mus) = (1 - \alpha) L (\betau - \betap) (1 - \mus) (1 - \yee(\mus))$.

Differentiating $w(\mus) = \beta_{\zidag}(\mus) \cdot \alpha_{\zidag}(\mus)$, we obtain
    \begin{align}
        \nabla_{\mus}(w) & = \nabla_{\mus}(\zidag) [\beta_{\zidag}(\mus) (1 - \alpha) (1 - \mus) + \nonumber
        \\ & \alpha_{\zidag}(\mus)(\betau - \betap)]  + \beta_{\zidag}(\mus) (1 - \alpha) (1 - \zidag(\mus)). \nonumber
    \end{align}
     Substituting $\nabla_{\mus}(\zidag)$ from \eqref{eq:del_zi_new} into the above equation yields
     \begin{align}
         \nabla&_{\mus}(w) = \Big[\beta_{\zidag}(\mus) (1 - \alpha) (1 - \mus) + \alpha_{\zidag}(\mus)(\betau - \betap)\Big] \nonumber
        \\ &\times \Big[-\gamma \nabla_{\mus}(w) \Big(\frac{\cp}{L(w(\mus) - \gamma)^2 (1 - \alpha) (\betau - \betap)} \nonumber
        \\ &+ \frac{\cu - \cp}{\gamma w(\mus) L_{\texttt{eq}}(\mus)}\Big)  - \frac{\cu - \cp}{(1 - \mus)L_{\texttt{eq}}(\mus)}\Big] \nonumber
        \\ &+ \beta_{\zidag}(\mus) (1 - \alpha) (1 - \zidag(\mus)). \nonumber
     \end{align}
By rearranging the terms, we can write $\nabla_{\mus}(w) = \frac{\mathcal{N}}{\mathcal{D}}$ where
\begin{align}
\mathcal{N} = &  -\Big[\beta_{\zidag}(\mus) (1 - \alpha) (1 - \mus) \nonumber
\\ & \quad + \alpha_{\zidag}(\mus)(\betau - \betap)\Big] \frac{\cu - \cp}{(1 - \mus) L_{\texttt{eq}}(\mus)} \nonumber
\\ & \quad + \beta_{\zidag}(\mus) (1 - \alpha) (1 - \zidag(\mus)), \nonumber
%\\ \Bigg/ \Bigg(1  \nonumber
\\ \mathcal{D} = &  1 + \Big[\beta_{\zidag}(\mus) (1 - \alpha) (1 - \mus) +\alpha_{\zidag}(\mus)(\betau - \betap)\Big] \nonumber
\\& \times \gamma \Big(\frac{\cp}{L(w(\mus) - \gamma)^2 (1 - \alpha) (\betau - \betap)} \nonumber
\\ & \qquad + \frac{\cu - \cp}{\gamma w(\mus) L_{\texttt{eq}}(\mus)}\Big).
\label{eq:del_ystar}
\end{align}

Note that the denominator \eqref{eq:del_ystar} is a positive quantity. Therefore, for $w(\mus)$ to be a strictly decreasing function of $\mus$, we must have $\mathcal{N} < 0$, or equivalently
    \begin{align}
& \beta_{\zidag}(\mus) (1 - \alpha) \Big( 1 - \zidag(\mus) - \frac{\cu - \cp}{L_{\texttt{eq}}(\mus)} \Big) \nonumber
        \\ & \qquad < \alpha_{\zidag}(\mus) (\betau - \betap) \frac{\cu - \cp}{(1 - \mus) L_{\texttt{eq}}(\mus)}.
        \label{eq:del_y_less_0}
    \end{align}

We now show that \eqref{eq:del_y_less_0} is satisfied under Assumption \ref{assump:cp_more}. Since, $\cp > (1 - \alpha) L (\betau - \gamma)$ holds and $\yee(\mus) < 1 - \frac{\gamma}{\betau}$, the following holds true:
    \begin{align}
        &\frac{\cp}{(1 - \alpha) L \yee(\mus)} > \betau \nonumber
        \\ \implies & \frac{\cp w(\mus)}{(1 - \alpha) L (w(\mus) - \gamma)} > \betau \nonumber
        \\ \implies & \frac{\betau}{\betau - \betap} - \frac{\cp w(\mus)}{(1 - \alpha) L (w(\mus) - \gamma) (\betau - \betap)}  < 0 \nonumber
        \\ \implies & 1 + \frac{\betap}{\betau - \betap} - \frac{\cp w(\mus)}{ (1 - \alpha) L (w(\mus) - \gamma)  (\betau - \betap)} \nonumber
        \\ & \qquad \qquad + \frac{\cu - \cp}{L_{\texttt{eq}}(\mus)} < \frac{\cu - \cp}{L_{\texttt{eq}}(\mus)} . 
        \label{eq:sufficient_assump}
    \end{align}
On careful examination of \eqref{eq:zi_new} and \eqref{eq:sufficient_assump}, we find that the L.H.S. of \eqref{eq:sufficient_assump} is equal to $1 - \zidag(\mus)$, which implies $1 - \zidag(\mus) < \frac{\cu - \cp}{L_{\texttt{eq}}(\mus)}$. Substituting this inequality into \eqref{eq:del_y_less_0}, we obtain the L.H.S. of \eqref{eq:del_y_less_0} is a negative quantity, whereas, the R.H.S. is a positive quantity, 

As a result, $\nabla_{\mus}(w) < 0$, and consequently, $\nabla_{\mus}(\yee) < 0$. Thus, the endemic infection level $\yee(\mus)$ is strictly decreasing in $\mus$ within the interval $\mus \in (\mus^{\min}, \mus^{\max})$. \end{proof}

Now we state our main result which reveals the optimal static persuasion signal. 

\begin{theorem}\label{theorem:oss}
\rev{Suppose Assumptions \ref{assumption:main} and \ref{assump:cp_more} hold.} Then, the optimal static signal which leads to the smallest proportion of infected individuals at the SNE is $\mu^\star_{\texttt{S(stat)}} = \mus^{\max}$.
\end{theorem}

\begin{proof}
It follows from Remark \ref{remark:assump_cp_more} that the SNEs described in the first two cases of Theorem \ref{lemma:sne_char} do not exist under Assumption \ref{assump:cp_more}. We now focus on the remaining equilibria. 

First, we observe that $\mus < \mus^{\min}$ cannot be the optimal signal, since it results in the SNE $(1-\frac{\gamma}{\betau},1,1)$, which has the highest infection level compared to other SNEs; note that the entire population remains unprotected in this SNE. 

Now, it follows from Lemma \ref{lemma:y_decreasing} that when $\mus$ increases from $\mus^{\min}$ till $\mus^{\max}$, the infected fraction $y_{\texttt{EE}}(1, \zibar^\dagger; \mus)$ is monotonically decreasing in $\mus$. Thus, the range $\mus \in (\mus^{\min},\mus^{\max})$ does not contain the optimal signal. 

Finally, at the SNE $(y_{\texttt{EE}}(1, 0; \mus),1,0)$, the infection level $$y_{\texttt{EE}}(1, 0; \mus) = 1 - \frac{\gamma}{\betap (\alpha + (1 - \alpha) \mus)}$$ is strictly increasing in $\mus \in [\mus^{\max}, 1]$. In other words, the infected proportion at the SNE is monotonically decreasing in $\mus$ when $\mus \in (\mus^{\min},\mus^{\max})$ and monotonically increasing when $\mus \in [\mus^{\max}, 1]$. 

\rev{It remains to be shown that the infected proportion at the SNE is continuous at $\mus^{\max}$. Inspecting Cases $3$ and $4$ of Theorem \ref{lemma:sne_char}, it suffices to show that as $\mu_{\mathtt{S}}$ approaches $\mu^{\max}_{\mathtt{S}}$ from below, $z^{\dagger}_{\bar{\mathtt{I}}}(\mu_{\mathtt{S}})$ approaches $0$. To this end, note from \eqref{eq:yee} and \eqref{eq:h,g} that the functions $y_{\mathtt{EE}}(1,z_{\bar{\mathtt{I}}};\mu_{\mathtt{S}})$ and $g(z_{\bar{\mathtt{I}}},\mu_{\mathtt{S}})$ are continuous in both of their arguments. Note further that $y_{\mathtt{EE}}(1,z_{\bar{\mathtt{I}}};\mu_{\mathtt{S}})$ is monotonically increasing in both $z_{\bar{\mathtt{I}}}$ and $\mu_{\mathtt{S}}$ when the other variable is held constant. Consequently, when $h(z_{\bar{\mathtt{I}}},\mu_{\mathtt{S}}) < 0$, $g(z_{\bar{\mathtt{I}}},\mu_{\mathtt{S}})$ is monotonically increasing in both $z_{\bar{\mathtt{I}}}$ and $\mu_{\mathtt{S}}$ when the other variable is held constant. Now suppose $\mu_{\mathtt{S}}^{\max}>0$, we have $g(0,\mu^{\max}_{\mathtt{S}}) = 0$. Consequently, we have $h(0,\mu_{\mathtt{S}}^{\max}) < 0$. Let $\bar{\mu}_{\mathtt{S}} = \mu_{\mathtt{S}}^{\max} - \epsilon$ for sufficiently small $\epsilon>0$. Then, we have $h(0,\bar{\mu}_{\mathtt{S}}) < 0$. Since $z^{\dagger}_{\bar{\mathtt{I}}}(\bar{\mu}_{\mathtt{S}})$ satisfies $g(z^{\dagger}_{\bar{\mathtt{I}}}(\mu^{\max}_{\mathtt{S}}-\epsilon),\mu^{\max}_{\mathtt{S}}-\epsilon) = 0$, it follows from the continuity and monotonicity of $g$ that $z^{\dagger}_{\bar{\mathtt{I}}}(\mu^{\max}_{\mathtt{S}}-\epsilon)$ is strictly positive and converges to $0$ as $\epsilon$ converges to $0$ from above. Therefore, the infected proportion at the endemic equilibrium $y_{\mathtt{EE}}(1,{z^{\dagger}_{\bar{\mathtt{I}}}(\mu_{\mathtt{S}})};\mu_{\mathtt{S}})$ converges to $y_{\mathtt{EE}}(1,0;\mu^{\max}_{\mathtt{S}})$ when $\mu_{\mathtt{S}}$ approaches $\mu^{\max}_{\mathtt{S}}$ from below.

The above analysis for $\mus \in [\mus^{\max}, 1]$} also subsumes the corner case where $\mus^{\max} = 0$, under which the SNE is $(y_{\texttt{EE}}(1, 0; \mus),1,0)$ for the entire range of $\mus \in [0,1]$. In this case, $(y_{\texttt{EE}}(1, 0; \mus),1,0)$ is monotonically increasing in $\mus$ for $\mus \in [0,1]$. Thus, $\mus^{\max}$ is the optimal static signal.
\end{proof}

\rev{The above result can be interpreted as the Stackelberg equilibrium where the principal acts as the leader who aims to minimize the infected proportion, and the agents act as followers who respond to the signal sent by the leader by playing the SNE strategy profile.} In Section \ref{sec:num_res}, we give two numerical examples to verify the above theorem. First, when Assumption \ref{assump:cp_more} is satisfied, we show that the infected proportion at the SNE is smallest when $\mu = \mus^{\max}$. We then show that when Assumption \ref{assump:cp_more} is violated, then $\mus^{\max}$ is not the optimal signal. 

\section{Optimal Dynamic Signal}
\label{sec:ods}

We now introduce the \emph{dynamic signaling scheme} where the sender is allowed to vary or modulate the persuasive signal over time depending on the states. To this end, we formulate a finite-horizon optimal control problem where the goal is to minimize the integral of the infected proportion $y(t)$ subject to the epidemic dynamics \eqref{eq:sis_protection} and evolutionary learning dynamics adopted by the individuals to update their strategies $\zsbar,\zibar$. Specifically, we adopt the Smith dynamics \cite{smith1984stability} to capture the evolution of the strategies since the stationary points of the Smith dynamics corresponds to the Nash equilibrium of the game, and vice versa \cite{sandholm2015population}. 

However, the Smith dynamics is not smooth due to the presence of a $\max$ operator. Consequently, we relax this operator with a soft-max function with parameter $\sigma \in [0, \infty)$; when $\sigma$ increases, the relaxed dynamics closely approximates the Smith dynamics. We now formally state the finite-time horizon optimal control problem as
\begin{align}
%\begin{aligned}
    \min_{\mus} &\int_{0}^{T} y(t) \,dt \nonumber \\ 
    \textrm{s.t.} \;\;  & \dot{y} = ((1 - y) [\betap + (\betau - \betap) \zibar] \nonumber
    \\ & \quad \times[\alpha + (1 - \alpha) (\zsbar \mus + \zibar (1 - \mus))] - \gamma) y, \nonumber 
    \\  & \dot{z}_{\bar{\Stb}} = \frac{1 - \zsbar}{1 + e^{(\sigma \Delta U[\bar{\Stb}])}}  - \frac{\zsbar}{1 + e^{(-\sigma \Delta U[\bar{\Stb}])}}, \nonumber 
     \\ & \dot{z}_{\bar{\Itb}} = \frac{1 - \zibar}{1 + e^{(\sigma \Delta U[\bar{\Itb}])}}  - \frac{\zibar}{1 + e^{(-\sigma \Delta U[\bar{\Itb}])}}, \nonumber
    % %\\\dot{z}_{\bar{\Itb}} = (1 - \zibar) \cdot \textrm{max} \{0, -\Delta U[\bar{\Itb}]\} - \zibar \cdot \textrm{max} \{0, \Delta U[\bar{\Itb}]\},
    \\ & 0 \leq \mus \leq 1,
    %\end{aligned}
    \label{eq:ocp}
\end{align}
where $e^{(\cdot)}$ is the exponential function and $\Delta U[\bar{\Stb}]$ and $\Delta U[\bar{\Itb}]$ are defined in \eqref{eq:utility}. We denote the optimal dynamic signal by $\mu^\star_{\texttt{S(dyn)}}(t)$ and suppress the dependence on $t$ for better readability.

%%%%%%%%%%%%%%%%%%%%%%%%%%%%%%%%%%%%%%%%%%%%%%%%%%%%%%%%%%%%%%%
%%%%%%%%%%%%%%%%%%%%%%%%%%%%%%%%%%%%%%%%%%%%%%%%%%%%%%%%%%%%%%%
%%%%%%%%%%%%%%%%%%%%%%%%%%%%%%%%%%%%%%%%%%%%%%%%%%%%%%%%%%%%%%%

\section{Numerical Results}
\label{sec:num_res}

We now compare the state trajectory under the static and dynamic signaling schemes via numerical simulations. The parameters whose values remain unaltered throughout this section are given in the following table.

\begin{center}
\resizebox{6cm}{!}{
\begin{tabular}{|c | c | c | c | c | c | c |} 
 \hline
 $\alpha$ & $\gamma$ & $L$ & $y(0)$ & $\zsbar(0)$ & $\zibar(0)$ & $\sigma$\\ [0.5ex] 
 \hline
 0.45 & 0.2 & 80 & 0.01 & 0.5 & 0.5 & 20\\ \hline
\end{tabular}
}
\end{center}

\rev{We choose the value of $L$ to be large enough to capture the adverse health and financial impacts of infection on an individual. We choose $\sigma = 20$, under which the dynamics of $z_{\bar{\Stb}}$ and $z_{\bar{\Itb}}$ closely approximate the Smith dynamics. In \cite{cheng2021face}, the authors showed that an ideal surgical mask reduces COVID-19 infection risk by approximately $70\%$. In order to model partial effectiveness, we choose $\alpha = 0.45$ which corresponds to $55\%$ reduction.}

\begin{figure}[ht!]
\centering
  \subfigure{\includegraphics[width=40mm]{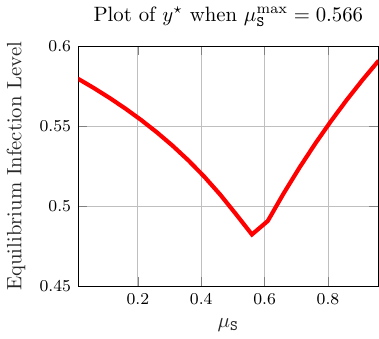}}
   \hspace{3mm}
  \subfigure{\includegraphics[width=40mm]{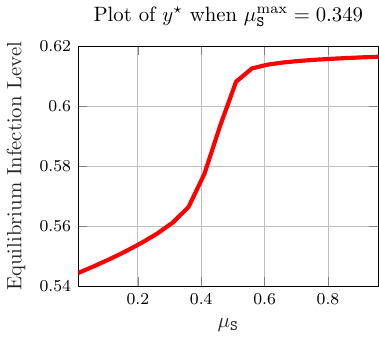}}
  \caption{Infected proportion at the SNE with respect to $\mus$ when Assumption \ref{assump:cp_more} is satisfied (left); and not satisfied (right).}
  \label{fig:y_vs_mus}
\end{figure}

\begin{figure*}[h!]
\centering
  \subfigure{\includegraphics[width=55mm]{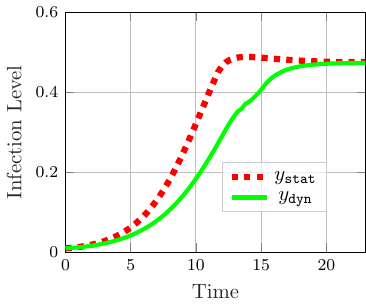}}
  \hspace{3mm}
  \subfigure{\includegraphics[width=56mm]{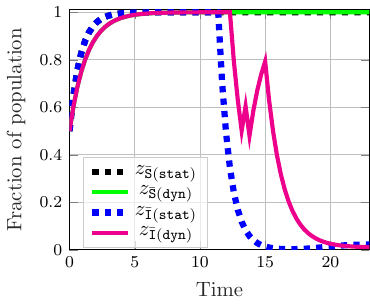}}
  \hspace{3mm}
  \subfigure{\includegraphics[width=56mm]{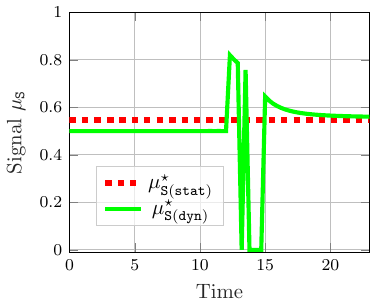}}
  \caption{Evolution of proportion of infected individuals (left), the proportions of individuals who remain unprotected (middle), and the signal $\mus$ (right) under both the static (dashed) and dynamic (solid) signaling schemes.}
  \label{fig:stat_vs_dyn}
\end{figure*}

\subsection{Static Signaling}\label{section:num_static}
Here, we vary $\mus$ from $\mus = 0.01$ to $\mus = 0.96$, and at each value of $\mus$, we obtain the infected proportion at the SNE. First, we choose the following parameter values: $\betau = 0.65$, $\betap = 0.5$, $\cp = 25$, $\cu = 32$, which satisfy Assumption \ref{assump:cp_more}. For these parameters, we have $\mus^{\min} = -2.028$ and $\mus^{\max} = 0.566$. The left panel of Figure \ref{fig:y_vs_mus} shows the infection level at the SNE with respect to $\mus$ under the above parameters. When $\mus \in [0.01, 0.566)$, the SNE is given by $(\yee(1, \zidag; \mus), 1, \zidag)$ following Theorem \ref{lemma:sne_char}. We observe that as $\mus$ is increases from $0.01$ to $0.566$, the equilibrium infection level $y^\star$ decreases. As $\mus$ grows beyond $\mus = \mus^{\max} = 0.566$, the SNE has the form $(\yee(1, 0; \mus), 1, 0)$, with $y^\star$ strictly increasing in $\mus$. Thus, under Assumption \ref{assump:cp_more}, the infected proportion at the SNE is smallest when $\mus = \mus^{\max}$. 

We now consider parameter values which violate Assumption \ref{assump:cp_more}. We choose $\betau = 0.9$, $\betap = 0.7$, $\cp = 19$, $\cu = 20$ under which $\mus^{\max} = 0.349$. The right panel of Figure \ref{fig:y_vs_mus} shows that the infected proportion at the SNE is not minimized at $\mus = \mus^{\max}$. In fact, the infected proportion is monotonically increasing with $\mus$.

\subsection{Dynamic Signaling}
We now illustrate the effectiveness of the optimal dynamic signaling scheme compared to the optimal static signaling scheme. We set the time horizon $T = 23$ and parameters $\betau = 0.65$, $\betap = 0.5$, $\cp = 20$, and $\cu = 25$ in accordance with Assumption \ref{assump:cp_more}. \rev{For the above parameters, the reproduction number is $2.5$ under protection and $3.25$ when protection is not adopted. These values fall within the empirically estimated range of the reproduction number for COVID-19 \cite[Figure 1]{he2020estimation}. }

For the above set of parameters, we obtain  $\mu^\star_{\texttt{S(stat)}} = \mus^{\max} = 0.548$, and determine the corresponding SNE following Theorem \ref{lemma:sne_char}. We then solve the optimal control problem \eqref{eq:ocp} via the numerical solver Quasi-Interpolation based Trajectory Optimization (QuITO) \cite{ganguly2023quito} which uses a direct multiple shooting (DMS) technique. 

We plot $y(t), \zsbar(t), \zibar(t)$, and the optimal control input under both static and dynamic signaling in Figure \ref{fig:stat_vs_dyn}. The dashed (respectively, solid) lines represent the quantities corresponding to the static (respectively, dynamic) signaling scheme. The rightmost plot in Figure \ref{fig:stat_vs_dyn} shows the optimal signal chosen by the sender under both static and dynamic schemes. While the signal is time-varying under the dynamic scheme, it converges to the optimal static signal $\mus^{\max}$ towards the end of the simulation. The middle panel of Figure \ref{fig:stat_vs_dyn} represents the evolution of strategies $\zsbar$ and $\zibar$, whereas the left panel plot shows the trajectory of the infected proportion under the static and dynamic signaling schemes. As expected, the infected proportion under the dynamic signaling scheme is smaller than the infected proportion under the static signaling scheme throughout the simulation; indeed the former aims to minimize the infection level along the entire trajectory while the latter is only concerned with the SNE. 

\rev{\begin{remark}
While we discuss numerical solutions of \eqref{eq:ocp} in this section, the nonlinearity in the dynamics of $y, \zsbar$, and $\zibar$ (particularly in the denominators in the dynamics of $\zsbar$ and $\zibar$), renders the analysis of the Lagrangian function extremely challenging. Thus, deriving closed-form analytical solutions of the optimal dynamic signal is beyond the scope of this paper, and remains a challenging open problem.  
\end{remark}}

%%%%%%%%%%%%%%%%%%%%%%%%%%%%%%%%%%%%%%%%%%%%%%%%%%%%%%%%%%%%%%%%%%%%%%%%
%%%%%%%%%%%%%%%%%%%%%%%%%%%%%%%%%%%%%%%%%%%%%%%%%%%%%%%%%%%%%%%%%%%%%%%%
%%%%%%%%%%%%%%%%%%%%%%%%%%%%%%%%%%%%%%%%%%%%%%%%%%%%%%%%%%%%%%%%%%%%%%%%

\subsection{Dynamic Signaling under Modified Cost Function}

Note that the value of $\mus^{\max}$ can potentially range from $0$ to $1$. However, when $\mus$ is close to or smaller than $0.5$, the credibility of the signaling scheme becomes questionable, and the individuals may eventually learn not to trust such a scheme. In order to mitigate this shortcoming, we modify the stage cost of the dynamic signaling to
$$J(t) = y(t) + c (1 - \mus)^2,$$
which encourages the sender to choose $\mus$ closer to $1$. The weight $c \in \mathbb{R}_{\geq 0}$ controls the trade-off between minimizing the infected proportion and keeping $\mus$ close to $1$.

\begin{figure}[h!]
\centering
  \subfigure{\includegraphics[width=40mm]{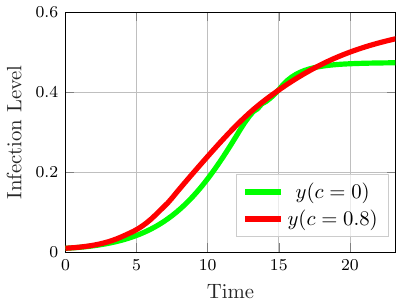}}
   \hspace{3mm}
  \subfigure{\includegraphics[width=40mm]{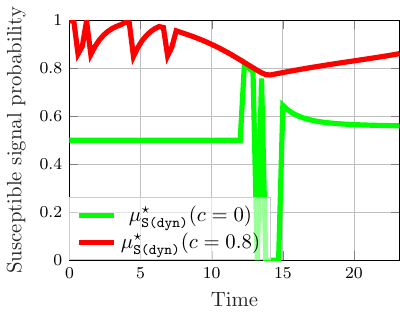}}
  \caption{Comparison of the infection levels and dynamic control signals for the original and modified stage cost.}
  \label{fig:comp_cost}
\end{figure}

Figure \ref{fig:comp_cost} compares the infected proportion and dynamic signal when $c = 0$ and $c=0.8$. As expected, the optimal input signal $\mu^\star_{\texttt{S(dyn)}}$ is closer to $1$ when $c =0.8$. However, the infected proportion is slightly larger under the modified cost compared to when $c = 0$.

Note that the trajectory of $\mu^\star_{\texttt{S(dyn)}}(t)$ shows a few sharp peaks and crests on the right panel of Figure \ref{fig:stat_vs_dyn} when $t\in[10,15]$. Similarly, the right panel of Figure \ref{fig:comp_cost} shows oscillatory behavior of $\mu^\star_{\texttt{S(dyn)}}(t)$ for $t \leq 7$. These are potentially due to the structure of the optimal control problem. In particular, if we minimize the corresponding Lagrangian functions, we obtain a quadratic (cubic) function of $\mus$ when the stage cost is $y(t)$ ($y(t) + c (1 - \mus)^2$), with possibly multiple local minima, which potentially gives rise to the fluctuations in the optimal control input.

\subsection{Relaxing the Assumption $\mu_{\mathtt{I}} = 1$}

\rev{Finally, we examine the consequence of relaxing the assumption $\mu_{\mathtt{I}} = 1$. To this end, we discretize both $\mus$ and $\mu_{\mathtt{I}}$ in steps of $0.005$ over the interval of $[0.01, 1]$. For each $\mu_{\mathtt{I}}$, we compute the infected proportions at the SNE over all possible values of $\mus$, denoted by $y_\mathtt{EE}(\mus,\mu_{\mathtt{I}})$ with a slight abuse of notation. On the left panel of Figure \ref{fig:vs_mui}, we plot the value of $\mus$ that results in the smallest infected proportion at the SNE, defined as $\mus^{\mathtt{opt}}(\mu_{\mathtt{I}}):=\argmin_{\mus} y_\mathtt{EE}(\mus,\mu_{\mathtt{I}})$, with respect to $\mu_{\mathtt{I}}$. The corresponding infected proportion ($\min_{\mus} y_\mathtt{EE}(\mus,\mu_{\mathtt{I}})$) is shown on the right panel for three combinations of $\betap$ and $\cp$ satisfying Assumptions \ref{assumption:main} and \ref{assump:cp_more}. All other parameter values are kept identical to those in Section \ref{section:num_static} that satisfy Assumption \ref{assump:cp_more}.

\begin{figure}[ht!]
\centering
  \subfigure{\includegraphics[width=38.5mm]{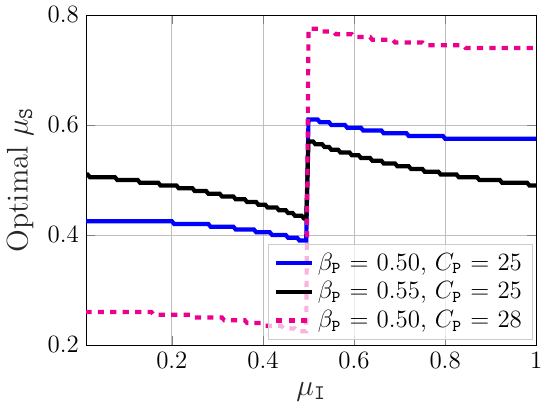}}
   \hspace{3mm}
  \subfigure{\includegraphics[width=41.5mm]{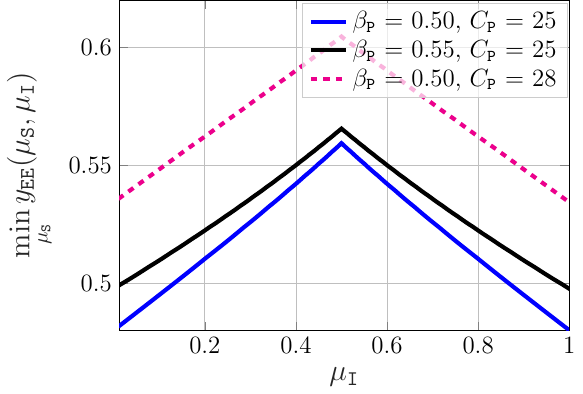}}
  \caption{\rev{Variation of $\mus^{\mathtt{opt}}(\mu_{\mathtt{I}})$ (left), and the infected proportion at the SNE under $(\mus^{\mathtt{opt}}(\mu_{\mathtt{I}}),\mu_{\mathtt{I}})$ (right) with respect to $\mu_{\mathtt{I}}$.}}
  \label{fig:vs_mui}
\end{figure}

%First, the left plot has some discrete steps, which arises due to discretization. We obtain smooth curves when the discretization parameter becomes infinitesimally small. 

As expected, the optimal $\mus$ at $\mu_{\mathtt{I}} = 1$ coincides with the corresponding $\mus^{\max}$. Now observe that as $\mu_{\mathtt{I}}$ decreases below $1$, the value of $\mus$ that achieves the minimum infection level, $\mus^{\mathtt{opt}}(\mu_{\mathtt{I}})$, does not change in a significant manner compared to $\mus^{\max}$. Thus, the optimal signal $\mus$ is somewhat robust with respect to $\mu_{\mathtt{I}}$. However, the magnitude of the minimum infection level ($\min_{\mus} y_\mathtt{EE}(\mus,\mu_{\mathtt{I}})$) increases as $\mu_{\mathtt{I}}$ decreases from $1$ to $0.5$. Furthermore, the plots exhibit symmetry around $\mu_{\mathtt{I}} = 0.5$, i.e., the signal combinations $(\mus^{\mathtt{opt}}, \mu_{\mathtt{I}})$ and $(1 - \mus^{\mathtt{opt}}, 1 - \mu_{\mathtt{I}})$ lead to the same infected proportion at the equilibrium. This symmetry arises because the posterior probabilities satisfy
$$\pi^{+}[\mathtt{s} \mid \bar{\Stb}]\big(\mus, \mu_{\mathtt{I}}\big) = \pi^{+}[\mathtt{s} \mid \bar{\Itb}]\big(1 - \mus, 1 - \mu_{\mathtt{I}} \big),$$ 
for $\mathtt{s} \in \{\Stb, \Itb\}$. Finally, we note that the smallest value of $y_\mathtt{EE}(\mus^{\mathtt{opt}}(\mu_{\mathtt{I}}),\mu_{\mathtt{I}})$ is obtained at $\mu_{\mathtt{I}}=1$. Establishing the optimality of $\mu_{\mathtt{I}}=1$ remains an important direction for future research.}  

\section{Conclusions}\label{sec:conc}
In this work, we investigated SIS epidemic containment via Bayesian persuasion of a large population of agents who strategically adopt a partially effective protection measure. We first derived the optimal static signal which minimizes the infected proportion at the stationary Nash equilibrium. We then formulated a finite-horizon optimal control problem to determine the optimal dynamic signaling scheme to minimize the infected proportion along the solution trajectory. Simulation results show that the dynamic signaling scheme is more effective in containing the infection prevalence over the entire trajectory. This work can be extended along several directions, including to networked epidemic models, and to settings where the parameters of the epidemic dynamics and cost functions of the players are not known to the sender. 

\section*{Acknowledgement}
The authors thank Prof. Ankur Kulkarni, Siddhartha Ganguly, and Abhisek Satapathi for helpful discussions.  

\bibliographystyle{IEEEtran}
\bibliography{main}

\end{document}